\pdfoutput=1

\documentclass{article}

\usepackage[a4paper,left=4cm]{geometry}	
\usepackage{lmodern,microtype}	
\usepackage[english]{babel}		
\usepackage{color}				

\definecolor{darkblue}{rgb}{0,0,.8}

\usepackage[nobreak]{cite} 

\usepackage{hyperref}
\hypersetup{
    colorlinks,
    citecolor=red,
    filecolor=black,
    linkcolor=darkblue,
    urlcolor=black
}

\usepackage{amsmath}
\usepackage{amssymb,amsthm}
\usepackage{bm}

\usepackage[capitalise,noabbrev]{cleveref}

\newtheorem{theorem}{Theorem}[section]
\newtheorem{lemma}[theorem]{Lemma}
\newtheorem{proposition}[theorem]{Proposition}

\renewcommand{\i}{\text{i}}
\newcommand{\Q}{\mathfrak{Q}}
\newcommand{\q}{\mathfrak{q}}
\newcommand{\ket}[1]{|#1\rangle}

\newcommand{\rrangle}{\rangle\hspace{-.05cm}\rangle}
\newcommand{\llangle}{\langle\hspace{-.05cm}\langle}

\title{\Large\bf On the transfer matrix of the supersymmetric eight-vertex model. I. Periodic boundary conditions}
\author{\normalsize \textsc{Christian Hagendorf} and \textsc{Jean Li\'enardy}\\
{\normalsize
  \begin{minipage}{\textwidth}
  \begin{center}
  \textit{
   Universit\'e catholique de Louvain\\
  Institut de Recherche en Math\'ematique et Physique\\
  Chemin du Cyclotron 2, 1348 Louvain-la-Neuve, Belgium} \\
    \medskip
  \href{mailto:christian.hagendorf@uclouvain.be}{\normalsize 
\texttt{christian.hagendorf@uclouvain.be}},
\href{mailto:jean.lienardy@uclouvain.be}{\normalsize 
\texttt{jean.lienardy@uclouvain.be}}
  \end{center}
  \end{minipage}
}
}

\date{\normalsize{\today}}

\begin{document}
\maketitle

\begin{abstract}
 The square-lattice eight-vertex model with vertex weights $a,b,c,d$ obeying the relation $(a^2+ab)(b^2+ab) = (c^2+ab)(d^2+ab)$ and periodic boundary conditions is considered. It is shown that the transfer matrix of the model for $L=2n+1$ vertical lines and periodic boundary conditions along the horizontal direction possesses the doubly degenerate eigenvalue $\Theta_n = (a+b)^{2n+1}$. This proves a conjecture by Stroganov from 2001. The proof uses the supersymmetry of a related XYZ spin-chain Hamiltonian.  The eigenstates of the transfer matrix corresponding to $\Theta_n$ are shown to be the ground states of the spin-chain Hamiltonian. Moreover, for positive vertex weights $\Theta_n$ is the largest eigenvalue of the transfer matrix.
\end{abstract}

\section{Introduction}
\label{sec:Introduction}

In this article, we investigate the transfer-matrix of the eight-vertex model on the square lattice with $L$ vertical lines and periodic boundary conditions along the horizontal direction \cite{baxterbook}. We focus on the case where the vertex weights $a,b,c,d$ are non-zero and related by
\begin{equation}
  (a^2+ab)(b^2+ab) = (c^2+ab)(d^2+ab).
  \label{eqn:CL8V}
\end{equation}
This special case of the eight-vertex model is known to be connected to a variety of topics such as elliptic solutions to functional equations \cite{baxter:89,fabricius:05,rosengren:16}, families of special solutions to the Painlev\'e VI equation \cite{bazhanov:05,bazhanov:06,rosengren:13,rosengren:13_2,rosengren:14,rosengren:15}, supersymmetry \cite{fendley:10,hagendorf:12,hagendorf:13} and combinatorics \cite{razumov:10,mangazeev:10,zinnjustin:13}. Because of its relation to supersymmetry, we follow Rosengren \cite{rosengren:15} and call the eight-vertex model with \eqref{eqn:CL8V} the \textit{supersymmetric eight-vertex model.}

Many of the results on the supersymmetric eight-vertex model rely on the existence of a remarkably simple eigenvalue of its transfer matrix. In fact, Stroganov \cite{stroganov:01,stroganov:01_2} conjectured in 2001 that for odd $L=2n+1,\, n\geqslant 0,$ the spectrum of the transfer matrix contains the doubly degenerate eigenvalue $\Theta_n = (a+b)^{2n+1}$. We refer to \cite{razumov:10} for a detailed description of how this was motivated by Baxter's work on the eight-vertex model. In the present article we prove Stroganov's conjecture for $n\geqslant 1$.  (The case $n=0$ is trivial.) To this end, we utilise a well-known relation between the transfer matrix of the eight-vertex model and the Hamiltonian of the XYZ spin chain. For $L$ sites and periodic boundary conditions, the spin-chain Hamiltonian is given by
\begin{equation}
   H_{\text{\tiny XYZ}} = -\frac{1}{2}\sum_{j=1}^L J_x \sigma_j^x\sigma_{j+1}^x +  J_y \sigma_j^y\sigma_{j+1}^y+ J_z \sigma_j^z\sigma_{j+1}^z, \quad \sigma^a_{L+1} = \sigma^a_1,\,a=x,y,z. 
  \label{eqn:XYZHam}
\end{equation}
Here, the $\sigma^a_j,\,a=x,y,z,$ are the standard Pauli matrices $\sigma^a$  acting on the site $j=1,\dots,L$. The real constants $J_x,J_y,J_z$ are the spin chain's anisotropy parameters.
For certain special choices of these parameters, which we discuss later, the transfer matrix of the eight-vertex model with generic weights commutes with the Hamiltonian \cite{sutherland:70,baxter:71,baxter:72_2}. They can therefore be simultaneously diagonalised. For the supersymmetric eight-vertex model, whose weights are related by \eqref{eqn:CL8V}, one such choice for the anisotropy parameters is given by
\begin{equation}
  J_x = 1+\zeta,\quad J_y = 1-\zeta,\quad J_z = \frac{1}{2}(\zeta^2-1),
  \label{eqn:CL}
\end{equation}
where
\begin{equation}
  \zeta = \frac{cd}{ab}.
\end{equation}
The XYZ Hamiltonian with anisotropy parameters \eqref{eqn:CL} possesses a lattice supersymmetry \cite{hagendorf:12}: Its restriction to a certain subsector of the spin-chain Hilbert space can be written, up to a constant, as the anticommutator of a nilpotent operator and its adjoint. These operators are called the supercharges. For odd $L=2n+1$, $n\geqslant 1$, the supersymmetry implies that there exists a two-dimensional space of special eigenstates of the 
Hamiltonian, the so-called supersymmetry singlets \cite{hagendorf:13}. We characterise this space and prove that it spans the space of the spin-chain ground states. Furthermore, we show that it is equal to the eigenspace of $\Theta_n$. We summarise our main results in the following theorem:
\begin{theorem}
  \label{thm:MainResult}
  For each $L=2n+1,\, n\geqslant 1,$ and non-zero vertex weights, the transfer matrix of the supersymmetric eight-vertex model possesses the doubly degenerate eigenvalue $\Theta_n = (a+b)^{2n+1}$. Its eigenspace is spanned by the ground states of the XYZ Hamiltonian \eqref{eqn:XYZHam} with $L=2n+1$ sites and the anisotropy parameters \eqref{eqn:CL} where $\zeta =cd/ab$.
\end{theorem}
\noindent
Moreover, we show that if the vertex weights are positive then $\Theta_n$ is the largest eigenvalue of the transfer matrix.

An important ingredient of the proof of \cref{thm:MainResult} is a commutation relation between the transfer matrix and the supercharges. This relation was conjectured in \cite{hagendorf:12}. Here, we prove it by using a special property of the $R$-matrix of the supersymmetric eight-vertex model. The observation of this property leads us to new insights about the relation between lattice supersymmetry and quantum integrability. 

The layout of this paper is as follows: In \cref{sec:SUSY}, we revisit the supersymmetry of the XYZ Hamiltonian with periodic boundary conditions and the anisotropy parameters \eqref{eqn:CL}. In particular, we reconsider and extend the proof of the existence of the supersymmetry singlets of the spin-chain Hamiltonian. Furthermore, we show that they are indeed the ground states of the XYZ Hamiltonian. We investigate the relation between the supercharges and the transfer matrix of the periodic eight-vertex model in \cref{sec:8V}. We use this relation to compute the action of the transfer-matrix on the space of supersymmetry singlets. This allows us to prove \cref{thm:MainResult}. In \cref{sec:Conclusion}, we present our conclusions and discuss some open problems.

\section{The XYZ spin chain and supersymmetry}
\label{sec:SUSY}
In this section, we review the supersymmetry of the periodic XYZ spin chain with the anisotropy parameters \eqref{eqn:CL}. We provide a short summary of basic notations and conventions in \cref{sec:XYZ}. Furthermore, we recall a few elementary symmetries of the spin-chain Hamiltonian. In \cref{sec:LatSUSY} we recall the definition of the supercharges. We discuss the existence of the supersymmetry singlets of the Hamiltonian in \cref{sec:Cohomology} by analysing the (co)homology of the supercharges. In \cref{sec:E0States} we characterise the supersymmetry singlets and prove that they span the space of the spin chain's ground states.

The presentation here below is essentially self-contained. Many definitions and statements that we use can be found in \cite{hagendorf:12,hagendorf:13,meidinger:14,hagendorf:17}.

\subsection{The XYZ Hamiltonian and its symmetries}
\label{sec:XYZ}
\paragraph{Hilbert space.} The XYZ spin chain, described by the Hamiltonian \eqref{eqn:XYZHam}, is a model of interacting spins $1/2$. The Hilbert space of a single spin $1/2$ is $\mathbb C^2$. We denote its canonical basis by
\begin{equation}
  |{\uparrow}\rangle = \begin{pmatrix} 1 \\ 0\end{pmatrix}, \quad
  |{\downarrow}\rangle = \begin{pmatrix} 0 \\ 1\end{pmatrix}.
\end{equation}
The Hilbert space of the spin chain with $L$ sites is given by $V^L = V_1\otimes V_2 \otimes \cdots \otimes V_L$ where $V_j=\mathbb C^2$ is a local copy of the single-spin Hilbert space. The canonical orthonormal basis of $V^L$ is given by the set of all states 
\begin{equation}
 |s_1s_2\cdots s_L\rangle = |s_1\rangle \otimes |s_2\rangle \otimes \cdots \otimes |s_L\rangle,
  \label{eqn:BasisVectors}
\end{equation}
where each $s_j$ is either $\uparrow$ (spin up) or $\downarrow$ (spin down). 
Throughout this article, we use the standard complex scalar product on $V^L$. We write $\langle \psi|\psi'\rangle$ for the scalar product of two states $|\psi\rangle,|\psi'\rangle \in V^L$, where $\langle \psi| = | \psi\rangle^\dagger$.

The spin operators on $\mathbb C^2$ are given by the standard Pauli matrices
\begin{equation}
  \sigma^x =
  \begin{pmatrix}
  0 & 1\\
  1 & 0
  \end{pmatrix}
  ,
  \quad
  \sigma^y =
  \begin{pmatrix}
  0 & -\i\\
  \i & 0
  \end{pmatrix}
  ,
  \quad
  \sigma^{z} =
  \begin{pmatrix}
  1 & 0\\
  0 & -1
  \end{pmatrix}.
\end{equation}
We denote by $\sigma_j^a,\,a=x,y,z$ and $j=1,\dots, L$, the matrix $\sigma^a$ acting on the $j$-th factor of the tensor product \eqref{eqn:BasisVectors}.

\paragraph{Symmetry operators.} The XYZ Hamiltonian \eqref{eqn:XYZHam} is Hermitian and therefore diagonalisable. Below, we focus on certain special eigenstates. The analysis of these eigenstates uses a few simple symmetries of the Hamiltonian that we discuss now.

We start this discussion by considering its invariance under translations. The translation operator $\mathcal S$ acts on the basis of $V^L$ according to
\begin{equation}
  \mathcal S| s_1\cdots s_{L-1} s_L\rangle = | s_L s_1\cdots s_{L-1}\rangle.
\end{equation}
The translation invariance of the Hamiltonian is expressed through the commutation relation
\begin{equation}
  [H_{\text{\tiny XYZ}},\mathcal S]=0.
\end{equation}
The operator $\mathcal S$ is unitary. Therefore, it is diagonalisable. The Hilbert space $V^L$ is the direct sum of the corresponding eigenspaces. In the following, we will be particularly interested in the eigenstates of $\mathcal S$ with eigenvalue $(-1)^{L+1}$. We follow the terminology of \cite{meidinger:14} and call them \textit{alternate-cyclic states}. We denote by $W^L$ the corresponding eigenspace.

Furthermore, we note that the Hamiltonian preserves the spin parity:
\begin{equation}
  [H_{\text{\tiny XYZ}},\mathcal P]=0,\quad \mathcal P = (-1)^L \sigma_1^z\sigma_2^z \cdots \sigma_L^z.
\end{equation}
Each basis state \eqref{eqn:BasisVectors} is an eigenstate of the spin-parity operator $\mathcal P$. The corresponding eigenvalue is the parity of the number of spins up. The spin-parity invariance of the Hamiltonian allows one to look for eigenstates of $H_{\text{\tiny XYZ}}$ in sectors where this parity is fixed to $+1$ or $-1$.

Finally, the Hamiltonian is invariant under spin reversal:
\begin{equation}
  [H_{\text{\tiny XYZ}},\mathcal R]=0,\quad \mathcal R = \sigma_1^x\sigma_2^x \cdots \sigma_L^x.
\end{equation}
The spin-parity and spin-reversal operators have the commutation relation $\mathcal R \mathcal P = (-1)^L \mathcal P \mathcal R$. In particular, they anticommute for odd $L$. This implies that each eigenvalue of $H_{\text{\tiny XYZ}}$ has an even degeneracy for chains of odd length.

\subsection{Lattice supersymmetry}
\label{sec:LatSUSY}
From now on, we focus on the case where the anisotropy parameters are given by \eqref{eqn:CL}. For this choice, the XYZ Hamiltonian possesses a lattice supersymmetry on the subspace of alternate-cyclic states \cite{hagendorf:12,hagendorf:13}. The corresponding symmetry operators are the so-called supercharge $\Q$ and its adjoint $\Q^\dagger$. 

\paragraph{Supercharges.} The supercharge $\Q$ is constructed from an operator $\q$ that we call the \textit{local supercharge}. Its action on the basis states of the single-spin Hilbert space is given by \cite{hagendorf:12}
\begin{equation}
  \q|{\uparrow}\rangle = 0, \quad \q|{\downarrow}\rangle = |{\uparrow\uparrow}\rangle - \zeta|{\downarrow\downarrow}\rangle,
  \label{eqn:DefLocalQ}
\end{equation}
where $\zeta$ is a real \textit{non-zero} parameter.\footnote{ The value $\zeta = \frac{cd}{ab}=0$ corresponds to the cases where $c=0$ or $d=0$, which we exclude.} Using $\q$, we define local operators $\q_0,\q_1,\dots,\q_L$ that map the Hilbert space of a chain of length $L$ to the Hilbert space of a chain of length $L+1$. For $j=1,\dots,L$, we set
\begin{subequations}
\begin{equation}
  \q_j = \underset{j-1}{\underbrace{1\otimes \cdots \otimes 1}}\otimes \q \otimes \underset{L-j}{\underbrace{1\otimes \cdots \otimes 1}}.
\end{equation}
Furthermore, we define
\begin{equation}
  \q_{0} = \mathcal S^{-1} \q_{1} \mathcal S = \mathcal S\q_L.
\end{equation}%
\label{eqn:DefineLocalQ}%
\end{subequations}%

The supercharge $\Q$ is a length-increasing operator\footnote{ In related works, such as \cite{hagendorf:13,meidinger:14}, the supercharge is often denoted by $\mathfrak Q_L$. The subscript indicates that it acts on $V^L$. We omit these subscripts for the supercharge (and other operators) in order to keep the notations as simple as possible. If necessary, we explicitly indicate which space $\Q$ acts on.} that maps $V^L$ to $V^{L+1}$ for each $L\geqslant 1$. We define it through its action on the eigenspaces of the translation operator $\mathcal S$ in $V^L$. On the eigenspace of alternate-cyclic states $W^L$, the supercharge acts as the alternating sum
\begin{equation}
  \Q= \sqrt{\frac{L}{L+1}}\sum_{j=0}^L (-1)^j\q_j.
  \label{eqn:DefQonWN}
\end{equation}
On every other eigenspace of the translation operator we define the supercharge to be zero. One checks \cite{hagendorf:13} that the supercharge maps $W^{L}$ to $W^{L+1}$. 

We define the adjoint of the supercharge $\Q^\dagger$ by means of the scalar product of the spin-chain Hilbert space. It satisfies
\begin{equation}
  \langle \psi|(\Q^\dagger|\phi\rangle) = \langle\phi|(\Q|\psi\rangle)^\ast
\end{equation}
for all $|\phi\rangle \in V^L,\,|\psi\rangle \in V^{L-1}$, $L\geqslant 2$. 
It follows from this definition that the action of the adjoint supercharge on the eigenspaces of the translation operator $\mathcal S$ in $V^L$ is non-zero only on $W^L$. Furthermore, $\Q^\dagger$ maps $W^L$ to $W^{L-1}$.

One can show that the supercharge and its adjoint are nilpotent operators:
\begin{equation}
  \Q^2 = 0, \quad (\Q^\dagger)^2 =0.
  \label{eqn:Nilpotency}
\end{equation}
This means that the operators $\Q^2:V^{L} \to V^{L+2},\, L\geqslant 1,$ and $(\Q^\dagger)^2:V^L \to V^{L-2}, \, L \geqslant 3,$ yield zero on every state of $V^L$. This is trivial on the subspace of $V^L$ spanned by the states that are not alternate-cyclic. Conversely, on $W^L$ it can be shown \cite{hagendorf:12} by using the definition of the local supercharge \eqref{eqn:DefLocalQ}.

\paragraph{Hamiltonian.} The supercharge and its adjoint allow us to define a Hamiltonian
\begin{equation}
  H = \Q\Q^\dagger +\Q^\dagger\Q.
  \label{eqn:HViaQ}
\end{equation}
This Hamiltonian is a length-preserving operator unlike $\Q$ and $\Q^\dagger$. It follows from their definition that the action of $H$ yields zero on all eigenspaces of the translation operator in $V^L$ that are not equal to the subspace of alternate-cyclic states $W^L$. Conversely, the restriction of $H$ to $W^L$ is non-trivial\cite{hagendorf:12}. Up to a multiple of the identity matrix it is equal to the Hamiltonian of the XYZ spin chain \eqref{eqn:XYZHam} with special anisotropy parameters. Indeed, we have
\begin{equation}
  H = H_{\text{\tiny XYZ}}-E_0 \quad \text{on }W^L,
  \label{eqn:XYZHamSusy}
\end{equation}
provided that $J_x,J_y,J_z$ are given by \eqref{eqn:CL} and $E_0$ is set to
\begin{equation}
  E_0 = -\frac{L}{4}(3+\zeta^2).
  \label{eqn:E0Value}
\end{equation}

The relation \eqref{eqn:XYZHamSusy} between $H$ and the XYZ Hamiltonian implies that $H$ commutes with the spin-parity and spin-reversal operators $[H,\mathcal P] = [H,\mathcal R]=0$. Furthermore, it follows from the nilpotency \eqref{eqn:Nilpotency} of both $\Q$ and $\Q^\dagger$ that the following commutation relations hold:
\begin{equation}
  H\Q = \Q H, \quad H\Q^\dagger = \Q^\dagger H.
  \label{eqn:CommRelQH}
\end{equation}
Hence the supercharges are symmetry operators: The Hamiltonian $H$ is supersymmetric. Since the Hamiltonians on the left- and right-hand sides of the equalities in \eqref{eqn:CommRelQH} act on the Hilbert spaces of spin chains whose length differs by one, this supersymmetry is called dynamic. 
Because of \eqref{eqn:XYZHamSusy} we conclude that the restriction of the XYZ Hamiltonian to $W^L$ with the anisotropy parameters \eqref{eqn:CL} has a dynamic lattice supersymmetry.

\subsection{(Co)homology}
\label{sec:Cohomology}
The definition \eqref{eqn:HViaQ} implies that $H$ is a Hermitian operator and can therefore be diagonalised. Its eigenvalues are non-negative. If it possesses the eigenvalue $E=0$, then the corresponding eigenstates are the ground states of this Hamiltonian. We call them \textit{supersymmetry singlets} or \textit{zero-energy states.} They are the non-zero solutions of
\begin{equation}
  \Q|\Psi\rangle = 0, \quad \Q^\dagger |\Psi\rangle =0.
  \label{eqn:QOnGS}
\end{equation}
It follows from the definition of $\Q$ and $\Q^\dagger$ that these equations have many trivial solutions. Indeed, all eigenstates of the translation operator $\mathcal S$ that are not alternate-cyclic are zero-energy states. In the following, we focus on the alternate-cyclic zero-energy states. If they exist, then they are the ground states of the XYZ Hamiltonian with the anisotropy parameters \eqref{eqn:CL}, restricted to $W^L$. The corresponding eigenvalue is $E_0 = -L(3+\zeta^2)/4$.

Proving the absence or the existence of alternate-cyclic zero-energy states is a non-trivial problem except for a few special values of the parameter $\zeta$. One such special value is $\zeta=1$, where the Hamiltonian reduces to
\begin{equation}
  H = \sum_{j=1}^L (1-\sigma_j^x\sigma_{j+1}^x)\quad \text{on}\quad W^L.
\end{equation}
Its diagonalisation is elementary. It reveals that $H$ possesses no alternate-cyclic zero-energy states for even $L=2n$. Conversely, for odd $L=2n+1$, the subspace of zero-energy states in $W^L$ is two-dimensional. One basis of this eigenspace is given by
\begin{subequations}
\begin{align}
  |\Phi_n\rangle = \frac{1}{2}(1+\mathcal P)\sum_{s_1=\uparrow,\downarrow}\cdots \sum_{s_{2n+1}=\uparrow,\downarrow}|s_1\cdots s_{2n+1}\rangle,\\ |\bar \Phi_n\rangle = \frac{1}{2}(1-\mathcal P)\sum_{s_1=\uparrow,\downarrow}\cdots \sum_{s_{2n+1}=\uparrow,\downarrow}|s_1\cdots s_{2n+1}\rangle.
\end{align}%
\label{eqn:DefPhiTrivial}%
\end{subequations}%
These basis states have a definite spin parity and can be mapped onto each other through spin reversal:
\begin{equation}
  \label{eqn:PROnPhi}
  \mathcal P|\Phi_n\rangle =  |\Phi_n\rangle, \quad \mathcal P|\bar \Phi_n\rangle = - |\bar \Phi_n\rangle, \quad \mathcal R|\Phi_n\rangle = |\bar \Phi_n\rangle.
\end{equation}

For generic values of $\zeta$, the explicit diagonalisation of the Hamiltonian $H$ is non-trivial. Nonetheless, it is possible to prove the absence or the existence of alternate-cyclic zero-energy states by means of the supersymmetry. A proof can be found in \cite{hagendorf:13}. For completeness, we revisit this proof and extend it here below. We start our discussion with the following simple observation:
\begin{lemma}
  \label{lem:Impossible}
  A zero-energy state cannot be in the image of $\Q$ or $\Q^\dagger$.
  \begin{proof}
  Let us prove that a zero-energy state cannot be in the image of $\Q$. By contradiction, we suppose that $|\Psi\rangle = \Q|\alpha\rangle$ is a zero-energy state. From \eqref{eqn:QOnGS}, we infer $\Q^\dagger\Q|\alpha\rangle=\Q^\dagger|\Psi\rangle=0$. We take the scalar product of this equality with $|\alpha\rangle$ and find $\langle \alpha|\Q^\dagger\Q|\alpha\rangle = ||\Q|\alpha\rangle||^2=0$, which implies $|\Psi\rangle = \Q|\alpha\rangle = 0$. This contradicts the definition of a zero-energy state, which has to be non-zero, and therefore proves the claim.
  
  The proof that a zero-energy state cannot be in the image of $\Q^\dagger$ is similar.
  \end{proof}
\end{lemma}

This lemma suggests that the space of zero-energy states could be related to the kernel of $\Q$ or $\Q^\dagger$ modulo their respective images. This is indeed the case. To explain this relation, we recall some facts from supersymmetric quantum mechanics \cite{witten:82,hori:03} and (co)homology theory \cite{loday:92,masson:08}.

For each $L\geqslant 2$ the space of zero-energy states in $W^L$ is in bijection with the quotient space
\begin{equation}
  \mathcal H^L = \frac{\text{ker}\{\Q:W^L\to W^{L+1}\}}{\text{im}\{\Q:W^{L-1}\to W^{L}\}}.
\end{equation}
It is sometimes useful to define $\mathcal H^1 = \text{ker}\{\Q:W^1\to W^{2}\}$. The direct sum $\bigoplus_{L=1}^\infty \mathcal H^L$ is then called the cohomology of the supercharge. The elements of $\mathcal H^L$ are equivalence classes. Any such equivalence class can be represented by a state $|\Phi\rangle \in \text{ker}\{\Q:W^L\to W^{L+1}\}$, called a \textit{representative}. Conversely, given a state $|\Phi\rangle$ that is annihilated by the supercharge, we denote by $[|\Phi\rangle]$ the corresponding equivalence class. Notice that $[|\Phi\rangle + \Q|\Phi'\rangle] = [|\Phi\rangle]$. Hence, representatives are not unique. Furthermore, one can show \cite{witten:82} that if $|\Phi\rangle$ is the representative of a non-zero element of $\mathcal H^L$ then there is a state $|\alpha\rangle\in W^{L-1}$ such that  
\begin{equation}
  |\Psi\rangle = |\Phi\rangle + \Q|\alpha\rangle
  \label{eqn:CohomRep}
\end{equation}
is an alternate-cyclic zero-energy state. Conversely, each alternate-cyclic zero-energy state $|\Psi\rangle$ can be decomposed as the sum \eqref{eqn:CohomRep} of a representative of a non-zero element of $\mathcal H^L$ and a state that is in the image of $\Q$.

The (formal) symmetry of the Hamiltonian under the exchange of the supercharge and its adjoint suggests that we could as well have considered
\begin{equation}
  \mathcal H_L = \frac{\text{ker}\{\Q^\dagger:W^L\to W^{L-1}\}}{\text{im}\{\Q^\dagger:W^{L+1}\to W^{L}\}}
\end{equation}
for $L\geqslant 2$. Indeed, the space of alternate-cyclic zero-energy states is isomorphic to $\mathcal H_L$ for each $L\geqslant 2$, too. Furthermore, we define $\mathcal H_1=W^1/\text{im}\{\Q^\dagger :W^2\to W^1\}$. The direct sum $\bigoplus_{L=1}^{\infty} \mathcal H_L$ is called the homology of the adjoint supercharge. The elements of $\mathcal H_L$  are also equivalence classes. They can be represented by states $|\Phi'\rangle\in\text{ker}\{\Q^\dagger:W^L\to W^{L-1}\}$. As above, we denote by $[|\Phi'\rangle]$ the equivalence class of such a state.\footnote{We use the same notation for the equivalence classes of $\mathcal H^L$ and $\mathcal H_L$.} One can show \cite{witten:82} that if $|\Phi'\rangle$ represents an non-zero element of $\mathcal H_L$ then there is a state $|\beta\rangle\in W^{L+1}$ such that
\begin{equation}
  |\Psi\rangle = |\Phi'\rangle + \Q^\dagger|\beta\rangle
  \label{eqn:HomRep}
\end{equation}
is an alternate-cyclic zero-energy state. Conversely, each alternate-cyclic zero-energy state $|\Psi\rangle$ can be decomposed as the sum \eqref{eqn:HomRep} of a representative of a non-zero element of $\mathcal H_L$ and a state that is in the image of $\Q^\dagger$. 

\paragraph{Conjugation.} We conclude that to prove the (non-)existence of alternate-cyclic zero-energy states of the Hamiltonian $H$ for $L$ sites, it is sufficient to find $\mathcal H^L$ or $\mathcal H_L$.  We now compute these spaces for each $L\geqslant 2$. In order to stress their dependence on $\zeta$, we write $\mathcal H^L = \mathcal H^L(\zeta)$ and $\mathcal H_L = \mathcal H_L(\zeta)$. 

Let us consider first the case where $\zeta = 1$. The explicit diagonalisation of the Hamiltonian shows that
\begin{subequations}
\begin{align}  
 &\mathcal H^{2n}(\zeta=1) = 0, \quad \mathcal H^{2n+1}(\zeta=1) = \mathbb C[|\Phi_n\rangle] \oplus \mathbb C[|\bar \Phi_n\rangle],\\
  &\mathcal H_{2n}(\zeta=1) = 0, \quad \mathcal H_{2n+1}(\zeta=1) = \mathbb C[|\Phi_n\rangle] \oplus \mathbb C[|\bar \Phi_n\rangle],
\end{align}%
  \label{eqn:DefHTrivial}%
\end{subequations}%
for each $n\geqslant 1$ where $|\Phi_n\rangle,\,|\bar \Phi_n\rangle$ are the states defined in \eqref{eqn:DefPhiTrivial}.

The corresponding results for generic values of $\zeta$ can be inferred from \eqref{eqn:DefHTrivial}. To this end, we introduce an operator $m(\lambda)$ whose action on the basis states of the single-spin Hilbert space is given by
\begin{equation}
  m(\lambda)|{\uparrow}\rangle = \lambda|{\uparrow}\rangle,
  \quad m(\lambda)|{\downarrow}\rangle = \lambda^2|{\downarrow}\rangle.
\end{equation}
The operator $m(\lambda)$ and the local supercharge $\q = \q(\zeta)$ satisfy the relation
\begin{equation}
  \left(m(\lambda)\otimes m(\lambda)\right)\q(\lambda^{-2} \zeta)  =  \q(\zeta)
   m(\lambda).
   \label{eqn:LocalConj}
\end{equation}
On $V^L$, we define the operator $ \mathcal M(\lambda) = m_1(\lambda)m_2(\lambda)\cdots m_L(\lambda)$ where $m_j(\lambda)$ is $m(\lambda)$ acting on the $j$-th factor of the tensor product \eqref{eqn:BasisVectors}. $\mathcal M(\lambda)$ preserves $W^L$ and is invertible for $\lambda\neq 0$. Let us write $\Q(\zeta)$ and $\Q(\zeta)^\dagger$ for the supercharge and its adjoint in order to stress their dependence on the parameter $\zeta$. It follows from \eqref{eqn:LocalConj} that
\begin{subequations}
\begin{align}
  \mathcal M(\lambda)\Q(\lambda^{-2}\zeta) &= \Q(\zeta)\mathcal M(\lambda),\\
  \mathcal M(\lambda^{-1})\Q(\lambda^{-2}\zeta)^\dagger &= \Q(\zeta)^\dagger\mathcal M(\lambda^{-1}).
\end{align}
  \label{eqn:ConjQ}%
\end{subequations}
Hence, for non-zero $\lambda$ one may relate the (adjoint) supercharges with parameters $\lambda^{-2}\zeta$ and $\zeta$ by conjugation with an invertible mapping. 
This conjugation property implies \cite{witten:82} that the following mappings are bijections:
\begin{subequations}
\begin{align}
  &\mathcal M^\sharp(\lambda):\mathcal H^{L}(\lambda^{-2}\zeta)\to \mathcal H^{L}(\zeta), \quad \mathcal M^{\sharp}(\lambda)[|\Phi\rangle] = [\mathcal M(\lambda)|\Phi\rangle],\\
  &\mathcal M_\sharp(\lambda):\mathcal H_{L}(\lambda^{-2}\zeta)\to \mathcal H_{L}(\zeta), \quad \mathcal M_{\sharp}(\lambda)[|\Phi'\rangle] = [\mathcal M(\lambda^{-1})|\Phi'\rangle].
\end{align}%
\label{eqn:Bijections}%
\end{subequations}%
The existence of these bijections was observed in \cite{hagendorf:13}. It implies that $\dim \mathcal H^L(\zeta)=\dim \mathcal H^L(\lambda^{-2}\zeta)$ and $\dim \mathcal H_L(\zeta)=\dim \mathcal H_L(\lambda^{-2}\zeta)$ for each $L\geqslant 1$. This allows one to compute the dimension of the space of alternate-cyclic zero-energy states as a function of the number of sites. Here, we extend the work of \cite{hagendorf:13} and use the bijections to explicitly compute $\mathcal H^{L}(\zeta)$ and $\mathcal H_{L}(\zeta)$ for non-zero $\zeta$. For $\zeta > 0$, we introduce the states
 \begin{equation}
  |\Phi_n(\zeta)\rangle = \zeta^{-(n+1)}\mathcal M(\zeta^{1/2})|\Phi_n\rangle, \quad  |\bar \Phi_n(\zeta)\rangle = \zeta^{-(n+1/2)}\mathcal M(\zeta^{1/2})|\bar \Phi_n\rangle.
  \label{eqn:DefPhi}
\end{equation}%
These states are polynomials in $\zeta$. Furthermore, we infer from \eqref{eqn:ConjQ} that they satisfy
\begin{subequations}
\begin{alignat}{2}
  & \Q(\zeta)|\Phi_n(\zeta)\rangle = 0, \quad && \Q(\zeta)|\bar \Phi_n(\zeta)\rangle = 0,\\
  & \Q^\dagger(\zeta)|\Phi_n(\zeta^{-1})\rangle = 0, \quad  && \Q^\dagger(\zeta)|\bar \Phi_n(\zeta^{-1})\rangle = 0.
\end{alignat}%
  \label{eqn:QOnPhi}%
\end{subequations}%
It follows from \eqref{eqn:Bijections} that for $\zeta>0$ we have \begin{subequations}
\begin{align}
  \mathcal H^{2n}(\zeta) &= 0, \quad \mathcal H^{2n+1}(\zeta) = \mathbb C[|\Phi_n(\zeta)\rangle] \oplus \mathbb C[|\bar \Phi_n(\zeta)\rangle], \\
  \mathcal H_{2n}(\zeta) &= 0, \quad \mathcal H_{2n+1}(\zeta) = \mathbb C[|\Phi_n(\zeta^{-1})\rangle] \oplus \mathbb C[|\bar \Phi_n(\zeta^{-1})\rangle].
\end{align}
\label{eqn:HomCohomQ}%
\end{subequations}%
The polynomiality of the states defined in \eqref{eqn:DefPhi} allows us to extend these relations to generic but non-zero values of $\zeta$.

Our construction of $\mathcal H^L(\zeta)$ and $\mathcal H_L(\zeta)$  clearly fails if $\zeta=0$ (which is the reason for requiring that $\zeta$ be non-zero). Indeed, in this case the conjugation relation \eqref{eqn:ConjQ} implies that the supercharges commute with the operator $\mathcal M(\lambda)$ for any finite $\lambda$. However, the commutation relation does not allow us to establish a relation between $\mathcal H^L(\zeta = 0)$ and $\mathcal H^L(\zeta=1)$, nor between $\mathcal H_L(\zeta = 0)$ and $\mathcal H_L(\zeta=1)$.
\subsection{Zero-energy states}
\label{sec:E0States}
We now use \eqref{eqn:HomCohomQ} in order to characterise the space of alternate-cyclic zero-energy states of the Hamiltonian $H$.
\begin{theorem}
\label{prop:CohomRep}
For each $n\geqslant 1$, the Hamiltonian \eqref{eqn:HViaQ} with $L=2n$ does not possess alternate-cyclic zero-energy states. If $L=2n+1$, then the space of alternate-cyclic zero-energy states is spanned by
  \begin{equation}
    |\Psi_n\rangle = \lambda_n|\Phi_n(\zeta)\rangle + \Q|\alpha_n\rangle,\quad |\bar\Psi_n\rangle = \bar \lambda_n|\bar \Phi_n(\zeta)\rangle + \Q|\bar \alpha_n\rangle,
    \label{eqn:CohomRepE0States}
  \end{equation}
where $|\alpha_n\rangle,|\bar \alpha_n\rangle \in W^{2n}$. The constants $\lambda_n,\bar \lambda_n$ are non-zero and given by
  \begin{equation}
    \lambda_n = \frac{1}{4^n}\langle \Phi_n(\zeta^{-1})|\Psi_n\rangle, \quad  \bar \lambda_n = \frac{1}{4^n}\langle \bar \Phi_n(\zeta^{-1})|\bar \Psi_n\rangle.
  \end{equation}
   \begin{proof}
    The absence and existence of the alternate-cyclic zero-energy states in $W^{2n}$ and $W^{2n+1}$, respectively, follow from \eqref{eqn:HomCohomQ}. If $L=2n+1$, then the decompositions \eqref{eqn:CohomRepE0States} are a consequence of \eqref{eqn:CohomRep}.
       
   The constants $\lambda_n,\bar \lambda_n$ have to be non-zero because otherwise the zero-energy states would be in the image of the supercharge, which is impossible because of \cref{lem:Impossible}. In order to find $\lambda_n$ we compute the scalar product
   \begin{equation}
     \langle \Phi_n(\zeta^{-1})|\Psi_n\rangle = \lambda_n  \langle \Phi_n(\zeta^{-1})|\Phi_n(\zeta)\rangle +  \langle \Phi_n(\zeta^{-1})|\Q|\alpha_n\rangle .
   \end{equation}
    The first term on the right-hand side of this equality is $4^n \lambda_n$. The second term vanishes because of \eqref{eqn:QOnPhi}. This leads to $\lambda_n=\frac{1}{4^n} \langle \Phi_n(\zeta^{-1})|\Psi_n\rangle$. The computation of $\bar \lambda_n$ is similar.
      \end{proof}
\end{theorem}
Next, we show that for generic $\zeta$ the zero-energy states $|\Psi_n\rangle$ and $|\bar \Psi_n\rangle$ have the same spin parity and transformation behaviour under spin reversal \eqref{eqn:PROnPhi} as for $\zeta=1$.
\begin{proposition}
\label{prop:Symmetries}
For each $n\geqslant 1$, the alternate-cyclic zero-energy states defined in \eqref{eqn:CohomRepE0States} satisfy
\begin{equation}
  \mathcal P|\Psi_n\rangle = +|\Psi_n\rangle, \quad  \mathcal P|\bar \Psi_n\rangle = -|\bar \Psi_n\rangle.
\end{equation}
    Furthermore, they can be normalised in such a way that
  $\mathcal R|\Psi_n\rangle = |\bar \Psi_n\rangle$.
  \begin{proof}
     First, we consider the action of the spin-parity operator on the zero-energy states. To this end, we notice that this operator anticommutes with the supercharge
     \begin{equation}
       \Q \mathcal P + \mathcal P \Q = 0.
       \label{eqn:ACPQ}
     \end{equation}
This follows from the definition of the local supercharge \eqref{eqn:DefLocalQ}. We use this relation to show that $\mathcal P|\Psi_n\rangle = +|\Psi_n\rangle$. A short calculation leads to
      \begin{equation}
      \mathcal P|\Psi_n\rangle -|\Psi_n\rangle =  -\Q(\mathcal P+1)|\alpha_n\rangle,
    \end{equation}
    where we used that $\mathcal P|\Phi_n(\zeta)\rangle = + |\Phi_n(\zeta)\rangle$. Since the Hamiltonian $H$ commutes with the spin-parity operator $\mathcal P$, the left-hand side of this equality, if non-zero, is a zero-energy state. The right-hand side is in the image $\Q$. \cref{lem:Impossible} states that this is not possible. Hence, both sides have to vanish. This leads to the desired result. The proof of $\mathcal P|\bar \Psi_n\rangle=-|\bar \Psi_n\rangle$ is similar.
    
    Second, the states $|\Psi_n\rangle,\,|\bar \Psi_n\rangle$ have thus opposite spin parity and span a two-dimensional eigenspace of the Hamiltonian. The Hamiltonian commutes with the spin-reversal operator. We conclude that $\mathcal R |\Psi_n\rangle = \rho_n|\bar \Psi_n\rangle$ and $\mathcal R |\bar \Psi_n\rangle = \rho_n^{-1}|\bar \Psi_n\rangle$ for a non-vanishing complex number $\rho_n$. It can be set to one by adjusting the normalisation of the states.
    \end{proof}
\end{proposition}

\begin{theorem}
\label{prop:HomRep}
For each $n\geqslant 1$, the alternate-cyclic zero-energy states $|\Psi_n\rangle$ and $|\bar \Psi_n\rangle$ can be written as
  \begin{equation}
    |\Psi_n\rangle = \mu_{n}|\Phi_n(\zeta^{-1})\rangle + \Q^\dagger|\beta_n\rangle,\quad |\bar\Psi_n\rangle = \bar \mu_{n}|\bar \Phi_n(\zeta^{-1})\rangle + \Q^\dagger|\bar \beta_n\rangle,
    \label{eqn:HomRepE0States}
  \end{equation}
  where $|\beta_n\rangle,|\bar \beta_n\rangle \in W^{2(n+1)}$. The constants $\mu_n$ and $\bar \mu_n$ are non-zero and given by
  \begin{equation}
    \mu_n = \frac{1}{4^n}\langle \Phi_n(\zeta)|\Psi_n\rangle, \quad \bar \mu_n = \frac{1}{4^n}\langle \bar \Phi_n(\zeta)|\bar \Psi_n\rangle.
  \end{equation}
  \begin{proof}
    We focus on the state $|\Psi_n\rangle$. It follows from the decomposition \eqref{eqn:HomRep} and from \eqref{eqn:HomCohomQ} that there are constants $\mu_n,\, \nu_n$ and a state $|\beta_n\rangle \in W^{2(n+1)}$ such that
    \begin{equation}
      |\Psi_n\rangle = \mu_n|\Phi_n(\zeta^{-1})\rangle + \nu_n|\bar \Phi_n(\zeta^{-1})\rangle + \Q^\dagger|\beta_n\rangle.
      \label{eqn:PsiDecomp}
    \end{equation}
 We act on both sides of this equality with the spin-parity operator and find
  \begin{equation}
      |\Psi_n\rangle = \mu_n|\Phi_n(\zeta^{-1})\rangle - \nu_n|\bar \Phi_n(\zeta^{-1})\rangle - \Q^\dagger \mathcal P|\beta_n\rangle.    \end{equation}
 The difference of these two equalities leads to 
   \begin{equation}
   2 \nu_n|\bar \Phi_n(\zeta^{-1})\rangle = -\Q^\dagger(1+\mathcal P)|\beta_{n}\rangle.
     \end{equation}
    We take the scalar product of both sides of this equality with $|\bar \Phi_n(\zeta)\rangle$. The scalar product with the right-hand side vanishes because of \eqref{eqn:QOnPhi}. On the left-hand side, we find $2^{2n+1} \nu_n$ and therefore have $\nu_n=0$. Finally, we determine the value of $\mu_n$ by taking the scalar product of both sides of \eqref{eqn:PsiDecomp} with $|\Phi_n(\zeta)\rangle$.
    
    The reasoning for $|\bar \Psi_n\rangle$ is similar.
    \end{proof}
\end{theorem}

The decomposition of a zero-energy state as the sum of a representative and a state in the image of the supercharge is not unique. We now determine an alternative decomposition for $|\bar \Psi_n\rangle$, which will be useful in \cref{sec:8V}.

\begin{proposition}
\label{prop:AlternativeCohomRep}
For each $n\geqslant 1$, the state $|\bar\Psi_n\rangle$ can be written
\begin{equation}
  |\bar \Psi_n\rangle = \nu_{n}|{\uparrow \cdots \uparrow}\rangle + \Q|\gamma_{n}\rangle
  \label{eqn:AlternativeCohomRep}
\end{equation}
for some state $|\gamma_{n}\rangle \in W^{2n}$. The constant $\nu_n$ is non-zero and given by $\nu_{n} = \langle \bar \Phi_{n}(\zeta^{-1})|\bar \Psi_{n}\rangle$.
\begin{proof}
 We show that for each $n\geqslant 1$ there is a linear combination of $|{\uparrow \cdots \uparrow}\rangle\in W^{2n+1}$ and $|\Phi_n(\zeta)\rangle$ that is in the image of the supercharge.
   
To see this, we notice that for each $n\geqslant 1$ the state $|{\uparrow \cdots \uparrow}\rangle\in W^{2n+1}$ is annihilated by $\Q$. This can be seen from the definition of the local supercharge \eqref{eqn:DefLocalQ}. It follows from \eqref{eqn:HomCohomQ} that there are constants $\eta_n,\bar \eta_n$ and a state $|\delta_n\rangle \in W^{2n}$ such that
 \begin{equation}
   |{\uparrow \cdots \uparrow}\rangle = \eta_n |\Phi_n(\zeta)\rangle + \bar \eta_n|\bar \Phi_n(\zeta)\rangle + \Q|\delta_n\rangle.
   \label{eqn:Eq1}
 \end{equation} 
 We act on both sides of this equality with the spin-parity operator $\mathcal P$, which leads to
 \begin{equation}
 -|{\uparrow \cdots \uparrow}\rangle = \eta_n|\Phi_n(\zeta)\rangle - \bar \eta_n|\bar \Phi_n(\zeta)\rangle - \Q\mathcal P|\delta_n\rangle.
 \label{eqn:Eq2}
 \end{equation}
 We take the sum of \eqref{eqn:Eq1} and \eqref{eqn:Eq2} and find
 \begin{equation}
 2 \eta_n|\Phi_n(\zeta)\rangle+\Q(1-\mathcal P)|\delta_n\rangle=0.
 \end{equation}
 The projection of this equality onto the state $|\Phi_n(\zeta^{-1})\rangle$ leads to $\eta_n = 0$ (and therefore $\Q(1-\mathcal P)|\delta_n\rangle=0$).
 Hence \eqref{eqn:Eq1} becomes 
  \begin{equation}
    |{\uparrow \cdots \uparrow}\rangle = \bar \eta_n|\bar \Phi_n(\zeta)\rangle+\Q|\delta_n\rangle.
  \end{equation}
  We take the scalar product of both sides of this equality with $|\bar \Phi_n(\zeta^{-1})\rangle$ and find $\bar \eta_n=\frac{1}{4^n}$.
  
  Finally, we combine this result with \eqref{eqn:CohomRepE0States} and find
  \begin{equation}
    |\bar \Psi_n\rangle = 4^n \bar\lambda_n|{\uparrow\cdots\uparrow}\rangle +\Q(|\bar \alpha_n\rangle-4^n\bar \lambda_n |\delta_n\rangle).
  \end{equation}
  This leads to \eqref{eqn:AlternativeCohomRep} with $\nu_n = 4^n\bar \lambda_n = \langle \bar \Phi_{n}(\zeta^{-1})|\bar \Psi_{n}(\zeta)\rangle$ and $|\gamma_n\rangle=|\bar \alpha_n\rangle-4^n\bar \lambda_n |\delta_n\rangle$. 
  \end{proof}
\end{proposition}

As was pointed out above, the states $|\Psi_n\rangle$ and $|\bar \Psi_n\rangle$ span the space of the ground states of $H_{\text{\tiny XYZ}}$ with $L=2n+1$ sites and the anisotropy parameters \eqref{eqn:CL}, \textit{restricted to $W^{2n+1}$}. The next theorem shows that they are in fact the ground states on the full Hilbert space $V^{2n+1}$. This follows from a generalisation of a classical result \cite{yang:66} by Yang and Yang on the ground state of the XXZ chain and an argument by Yang and Fendley \cite{yang:04}.
\begin{theorem}
  \label{thm:GS}
  For each $L=2n+1,\, n\geqslant 1$, and non-zero $\zeta$ the states $|\Psi_n\rangle$ and $|\bar \Psi_n\rangle$ span the space of the ground states of $H_{\textnormal{\tiny XYZ}}$ with the anisotropy parameters defined in \eqref{eqn:CL}. The corresponding ground-state eigenvalue is $E_0 = - (2n+1)(3+\zeta^2)/4$.
  \begin{proof}
    We divide the proof into four steps.
        
    First, we notice that it is sufficient to prove the statement for $\zeta > 0$. The reason is that the Hamiltonians for $\zeta>0$ and $\zeta < 0$ can be related by a unitary transformation. Indeed, writing $H_{\text{\tiny XYZ}} = H_{\text{\tiny XYZ}}(\zeta)$, we have
    \begin{equation}
      H_{\text{\tiny XYZ}}(-\zeta) = \mathcal M(\i) H_{\text{\tiny XYZ}}(\zeta) \mathcal M(\i)^\dagger
    \end{equation}
    where $\mathcal M(\i)$ is the operator introduced in \cref{sec:Cohomology}. Furthermore, let us write $|\Psi_n(\zeta)\rangle$ and $|\bar \Psi_n(\zeta)\rangle$ for the alternate-cyclic zero-energy states. Using \cref{prop:CohomRep,prop:Symmetries} one can show that $\mathcal M(\i)|\Psi_n(\zeta)\rangle = \gamma_n|\Psi_n(-\zeta)\rangle$ and $\mathcal M(\i)|\bar \Psi_n(\zeta)\rangle = \bar \gamma_n|\bar \Psi_n(-\zeta)\rangle$ where $\gamma_n,\,\bar\gamma_n$ are non-zero complex numbers. Hence, if $\ket{\Psi_n(\zeta)}$ and $\ket{\bar\Psi_n(\zeta)}$ span the space of the ground states of $H_{\text{\tiny XYZ}}(\zeta)$, then $\ket{\Psi_n(-\zeta)}$ and $\ket{\bar\Psi_n(-\zeta)}$ will span the space of ground states of $H_{\text{\tiny XYZ}}(-\zeta)$ as well.
      
  Second, for $\zeta > 0$ the off-diagonal entries of $ H_{\text{\tiny XYZ}}$ are zero or negative. Hence, there is a constant $\lambda $ such that $\lambda - H_{\text{\tiny XYZ}}$ is a non-negative matrix with positive   diagonal entries. We consider the restriction $H_\pm$ of $\lambda - H_{\text{\tiny XYZ}}$ to the eigenspace of the spin-parity operator $\mathcal P$ associated to the eigenvalue $\pm 1$. The matrix $H_\pm$ is Hermitian and thus has real eigenvalues. Furthermore, the repeated action of $H_\pm$ on any basis state $|s_1s_2\cdots s_L\rangle$ with spin parity $\pm 1$ leads to linear combinations of basis states that have the same spin parity. The coefficients of these linear combinations are positive. Any other basis state $|s_1's_2'\cdots s_L'\rangle$ with this spin parity can be found in one of these linear combinations. Following \cite{yang:66}, we conclude that there exists an integer $m \geqslant 1$ such that $H_\pm^m$ is a positive matrix.  Hence, $H_\pm$ is \textit{irreducible and non-negative} \cite{meyer:00}. We may thus apply the \textit{Perron-Frobenius theorem for irreducible non-negative matrices}. It implies that the largest eigenvalue $\lambda_\pm$ of $H_\pm$ is non-degenerate. Furthermore, there is a unique state $|\Psi_\pm\rangle$ with positive components and norm one such that
\begin{equation}
  H_\pm |\Psi_\pm\rangle = \lambda_{\pm}|\Psi_\pm\rangle.
  \label{eqn:DefPsiPm}
\end{equation}
Considered as a vector of $V^{2n+1}$, $|\Psi_\pm\rangle$ has non-negative components. It spans the one-dimensional space of the ground states of $H_{\text{\tiny XYZ}}$ in the subsector where the spin parity is fixed to $\pm1$.

Third, following \cite{yang:04} we prove that $|\Psi_\pm\rangle$ is invariant under translations. Indeed, because of $[H_{\text{\tiny XYZ}},\mathcal S]=0$ and $[\mathcal P,\mathcal S]=0$, we have
\begin{equation}
  \mathcal S|\Psi_\pm\rangle = t_\pm|\Psi_\pm\rangle, \quad \text{with}\quad  t^{2n+1}_\pm=1.
\end{equation}
We now take the complex conjugate of this equation. Since the components of $|\Psi_\pm\rangle$ are real, we immediately find $\bar t_\pm = t^{-1}_\pm = t_\pm$ and therefore $t_\pm=1$. The state $|\Psi_\pm\rangle$ is therefore alternate-cyclic.

Finally, it follows from \cref{prop:CohomRep,prop:Symmetries} that $|\Psi_+\rangle$ and $|\Psi_-\rangle$ are proportional to $|\Psi_n\rangle$ and $|\bar \Psi_n\rangle$, respectively. The ground-state eigenvalue follows from \eqref{eqn:XYZHamSusy} and therefore is doubly degenerate.
\end{proof}
\end{theorem}

\section{The transfer-matrix eigenvalue}
\label{sec:8V}
In this section, we prove \cref{thm:MainResult}. In \cref{sec:TM}, we briefly review the definition and a few properties of the transfer matrix of the eight-vertex model. This is followed by a sequence of intermediate results for the supersymmetric eight-vertex model which we need for the proof of our main result. In \cref{prop:Restriction}, we show that the eigenstates corresponding to $\Theta_n= (a+b)^{2n+1}$ are necessarily alternate-cyclic zero-energy states of the Hamiltonian $H$. In \cref{sec:TMSUSY}, we prove \cref{prop:TQRelation} which establishes a commutation relation between the supercharge and the transfer matrix of the supersymmetric eight-vertex model. Furthermore, we show in \cref{prop:MatrixElementOp} that the commutation relation can be used to greatly simplify the computation of the action of the transfer matrix on the space of alternate-cyclic zero-energy states, \textit{even without the explicit knowledge of the states' components}. The actual computation is done in \cref{sec:EV}. In particular, in \cref{prop:Theta} we reduce the evaluation of the transfer-matrix eigenvalues on this subspace to a combinatorial problem. Eventually, the proof of our main theorem amounts to a combination of these propositions.

\subsection{The transfer matrix of the eight-vertex model}
\label{sec:TM}

The $R$-matrix of the eight-vertex model is an operator $R:\mathbb C^2 \otimes \mathbb C^2\to \mathbb C^2 \otimes \mathbb C^2$. In the standard basis $|{\uparrow\uparrow}\rangle,|{\uparrow\downarrow}\rangle,|{\downarrow\uparrow}\rangle,|{\downarrow\downarrow}\rangle$ of $\mathbb C^2\otimes \mathbb C^2$ it reads
\begin{equation}
  R = \begin{pmatrix}
    a & 0 & 0 & d\\
    0 & b & c & 0\\
    0 & c & b & 0\\
    d & 0 & 0 & a
  \end{pmatrix}.
\end{equation}
The transfer matrix of the eight-vertex model on the square lattice with $L$ vertical lines with periodic boundary conditions along the horizontal direction is an operator $\mathcal T:V^L \to V^L$, defined as
\begin{equation}
  \mathcal T = \text{tr}_{0}\left(R_{0L}R_{0L-1}\cdots R_{01}\right).
  \label{eqn:TM8V}
\end{equation}
Here $R_{ij}$ is the $R$-matrix acting non-trivially only the factors $V_i$ and $V_j$ in the product space $V_0 \otimes V^L= V_0\otimes V_1\otimes \cdots \otimes V_L$. The trace is taken over the space $V_0$. The transfer matrix is invariant under translations, spin reversal and preserves the spin parity \cite{baxterbook}. This is expressed by the following commutation relations
\begin{equation}
  [\mathcal T,\mathcal S]=[\mathcal T,\mathcal R] = [\mathcal T,\mathcal P] = 0.
  \label{eqn:CRTSymOps}
\end{equation}
Furthermore, one can show \cite{baxterbook} that $\mathcal T$ commutes with its transpose $[\mathcal T,\mathcal T^t]=0$. Hence it is a normal matrix and therefore diagonalisable by means of a unitary transformation.

To investigate in more detail the properties of the transfer matrix, it is often convenient to write the vertex weights of the model in terms of Jacobi theta functions \cite{whittaker:27}:
\begin{subequations}
\begin{align}
  a(u) = \rho \vartheta_4(2\eta,p^2)\vartheta_1(u+2\eta,p^2)\vartheta_4(u,p^2),\\
  b(u) = \rho \vartheta_4(2\eta,p^2)\vartheta_4(u+2\eta,p^2)\vartheta_1(u,p^2),\\
  c(u) = \rho \vartheta_1(2\eta,p^2)\vartheta_4(u+2\eta,p^2)\vartheta_4(u,p^2),\\
  d(u) = \rho \vartheta_1(2\eta,p^2)\vartheta_1(u+2\eta,p^2)\vartheta_1(u,p^2).
\end{align}%
\label{eqn:EllipticParametrisation}%
\end{subequations}%
Here, $\rho$ is a normalisation constant, $\eta$ the so-called crossing parameter, $p$ the elliptic nome and $u$ the spectral parameter. With this parameterisation, the $R$-matrix $R=R(u)$ satisfies the Yang-Baxter equation: $R_{12}(u-v)R_{13}(u)R_{23}(v)= R_{23}(v)R_{13}(u)R_{12}(u-v)$ for all $u,v$. It follows
from the Yang-Baxter equation that transfer matrices with different spectral parameters commute \cite{baxterbook,baxter:72}: Writing $\mathcal T = \mathcal T(u)$, we have
\begin{equation}
  [\mathcal T(u),\mathcal T(v)]=0, \quad \text{for all }u,v.
  \label{eqn:CommTMs}
\end{equation}
Hence, $\mathcal T(u)$ possesses an eigenbasis that is independent of the spectral parameter $u$.

The commutation relation \eqref{eqn:CommTMs} implies that the transfer matrix $\mathcal T(u)$ commutes with the Hamiltonian $H_{\text{\tiny XYZ}}$ of the XYZ spin chain with certain anisotropy parameters \cite{sutherland:70,baxter:71,baxter:72}. Indeed, a standard calculation leads to
\begin{equation}
  \mathcal T(0) = a(0)^L \mathcal S, \quad  \mathcal T(0)^{-1}\mathcal T'(0) = \frac{L(a'(0)+c'(0))}{2a(0)}-\frac{b'(0)}{a(0)} H_{\text{\tiny XYZ}}.
  \label{eqn:TUZero}
\end{equation}
Here, the anisotropy parameters of the spin-chain Hamiltonian are given by
\begin{equation}
  J_x = 1+\frac{d'(0)}{b'(0)}, \quad J_y = 1-\frac{d'(0)}{b'(0)}, \quad J_z = \frac{a'(0)-c'(0)}{b'(0)}.
  \label{eqn:Js}
\end{equation}
Using \eqref{eqn:EllipticParametrisation} (and a few identities between the Jacobi theta functions \cite{whittaker:27}) one finds $J_x = 1+\zeta$ and $J_y = 1-\zeta$ with
\begin{equation}
  \zeta = \left(\frac{\vartheta_1(2\eta,p^2)}{\vartheta_4(2\eta,p^2)}\right)^2= \frac{c(u) d(u)}{a(u)b(u)} ,
  \label{eqn:Zeta}
\end{equation}
and
\begin{equation}
J_z  = \frac{\vartheta_2(2\eta,p^2)\vartheta_3(2\eta,p^2)\vartheta_4(0,p^2)^2}{\vartheta_2(0,p^2)\vartheta_3(0,p^2)\vartheta_4(2\eta,p^2)^2}= \frac{a(u)^2+b(u)^2-c(u)^2-d(u)^2}{2 a(u)b(u) }.
\label{eqn:Jz}
\end{equation}
It follows from \eqref{eqn:CommTMs} that $[\mathcal T(u),H_{\text{\tiny XYZ}}]=0$ for these anisotropy parameters. Hence the transfer matrix and the Hamiltonian can be simultaneously diagonalised.

From now on, we consider the case where the anisotropy parameters of the XYZ chain are parameterised according to \eqref{eqn:CL}. It follows from \eqref{eqn:Zeta} and \eqref{eqn:Jz} that this parameterisation is equivalent to the relation \eqref{eqn:CL8V}, which defines the supersymmetric eight-vertex model. (Furthermore, it corresponds to the value $\eta = \pi/3$ of the crossing parameter.)

Our goal is to prove \cref{thm:MainResult} about the existence of a special transfer-matrix eigenvalue in this case. In the next proposition, we show that if the corresponding eigenvalue problem possesses non-zero solutions then they are necessarily alternate-cyclic zero-energy states of the Hamiltonian $H$. 
\begin{proposition}
  \label{prop:Restriction}
  Let $n\geqslant 1$ and suppose that $|\Psi\rangle\in V^{2n+1}$ is a non-zero solution of the eigenvalue equation
  \begin{equation}
  \mathcal T|\Psi\rangle = (a+b)^{2n+1}|\Psi\rangle,
  \label{eqn:EVEqu}
\end{equation}
  then we have
  \begin{equation}
    \mathcal S|\Psi\rangle = |\Psi\rangle, \quad H_{\text{\tiny \rm XYZ}}|\Psi\rangle =  E_0|\Psi\rangle
  \end{equation}
  with $E_0 = -(2n+1)(3+\zeta^2)/4$ where $\zeta=cd/ab$.
  \begin{proof}
    We use the parameterisation \eqref{eqn:EllipticParametrisation} and thus consider the eigenvalue problem
    \begin{equation}
      \mathcal T(u)|\Psi\rangle = (a(u)+b(u))^{2n+1}|\Psi\rangle
      \label{eqn:TUEVProblem}
    \end{equation}
    for fixed $n\geqslant 1$. Let us suppose that this equation has a non-zero solution $|\Psi\rangle\in V^{2n+1}$. Because of \eqref{eqn:CommTMs} we may suppose without loss of generality that it is independent of the spectral parameter $u$. For $u=0$, we have $\mathcal T(0)|\Psi\rangle = a(0)^{2n+1}|\Psi\rangle$ with $a(0)\neq 0$. It follows from \eqref{eqn:TUZero} that $\mathcal S|\Psi\rangle= |\Psi\rangle$. Next, we differentiate both sides of \eqref{eqn:TUEVProblem} with respect to $u$ and set $u=0$. Using \eqref{eqn:TUZero} we obtain
    \begin{equation}
      H_{\text{\tiny XYZ}}|\Psi\rangle = -(2n+1)\left(1+\frac{a'(0)-c'(0)}{2b'(0)}\right)|\Psi\rangle.
    \end{equation}
    On the right-hand side, we recognise the expression of $J_z$ given in \eqref{eqn:Js}. Since $J_z = (\zeta^2-1)/2$ we find $H_{\text{\tiny XYZ}}|\Psi\rangle=E_0|\Psi\rangle$.
      \end{proof}
\end{proposition}

\subsection{Transfer matrix and supersymmetry}
\label{sec:TMSUSY}

It was conjectured in \cite{hagendorf:12} that the transfer matrix of the supersymmetric eight-vertex model and the supercharges of the XYZ Hamiltonian with $\zeta = cd/ab$ have a simple commutation relation. Here, we prove this conjecture. The proof relies on a relation between the $R$-matrix of the supersymmetric eight-vertex model and the local supercharge and therefore sheds some light on the connection between integrability and supersymmetry. We refer to \cite{weston:17} for a recent investigation of this connection in the case of the six-vertex model.

We introduce an operator $A:\mathbb C^2 \to \mathbb C^2 \otimes \mathbb C^2$. Its action on the basis states $|{\uparrow}\rangle$ and $|{\downarrow}\rangle$ is 
\begin{equation}
  A|{\uparrow}\rangle = d\left(-\frac{c}{a}|{\uparrow\downarrow}\rangle + |{\downarrow\uparrow}\rangle\right), \quad A|{\downarrow}\rangle = c\left(|{\uparrow\uparrow}\rangle - \frac{d}{b}|{\downarrow\downarrow}\rangle\right).
\end{equation}
On $V_0 \otimes V^L$, we define
\begin{equation}
  A_0^1 = A \otimes \underset{L}{\underbrace{1 \otimes \cdots \otimes 1}}, \quad A_0^2 = \mathcal S A_0^1\mathcal S^{-1}.
\end{equation}
Here and in the following, the translation operator leaves $V_0$ unchanged and only acts on $V^L$.
The operators $A_0^1$ and $A_0^2$ allow us to establish a relation between the $R$-matrix of the eight-vertex model and the local supercharge whose proof is a straightforward computation.
\begin{lemma}
  We have the equality
  \begin{equation}
    R_{02}R_{01}\q_1 + (a+b) \q_1 R_{01} = A_{0}^2 R_{01} +R_{02}A_0^1
    \label{eqn:RMatrixQ}
  \end{equation}
  if and only if the relations \eqref{eqn:CL8V} hold.
 \end{lemma}
 By means of this lemma, we prove the following commutation relation between the supercharge and the transfer matrix of the eight-vertex model:
\begin{proposition}
  \label{prop:TQRelation}
The transfer matrix of the supersymmetric eight-vertex model obeys the commutation relation 
  \begin{equation}
    \mathcal T \Q+(a+b)\Q \mathcal T = 0
    \label{eqn:TQRelation}
  \end{equation}
  on $V^L$ for each $L\geqslant 1$.
  \begin{proof}
We notice that this relation trivially holds on any eigenspace of the translation operator that is not equal to the space of alternate-cyclic states.
  
    It is therefore sufficient to prove the relation on $W^L$. To this end, we multiply the relation \eqref{eqn:RMatrixQ} on the left by the product of $R$-matrices $R_{0L+1}\cdots R_{03}$ and take the trace over the space $V_0$. Using the identity $R_{0j}\q_1 = \q_1 R_{0j-1}$ for each $j=3,\dots,L+1$, we obtain
    \begin{equation}
      \mathcal T \q_1 +(a+b)\q_1 \mathcal T = \text{tr}_0\left(R_{0L+1}\cdots R_{03}A_{0}^2 R_{01}\right)+\text{tr}_0\left(R_{0L+1}\cdots R_{03}R_{02}A_{0}^1\right).
      \label{eqn:Intermediate}
    \end{equation}
    We define the operator $\mathcal A:V^L\to V^{L+1}$ as 
    \begin{equation}
      \mathcal A =  \text{tr}_0\left(R_{0L+1}\cdots R_{03}R_{02}A_{0}^1\right),
    \end{equation}
     and rewrite \eqref{eqn:Intermediate} in terms of $\mathcal A$ and the translation operator $\mathcal S$. Using the relation $\mathcal S A_0^1 \mathcal S^{-1} = A_0^2$ and the cyclic property of the trace operation, we find
    \begin{equation}
      \mathcal T \q_1 +(a+b)\q_1 \mathcal T = \mathcal S\mathcal A \mathcal S^{-1} +\mathcal A.
    \end{equation}
    By conjugation with $\mathcal S^{j-1}$ this equality generalises to 
     \begin{equation}
      \mathcal T \q_j +(a+b)\q_j \mathcal T = \mathcal S^j\mathcal A \mathcal S^{-j} +\mathcal S^{j-1}\mathcal A \mathcal S^{-(j-1)}, \quad j=0,\dots,L.
    \end{equation}
    Here, we used \eqref{eqn:DefineLocalQ} and the commutation relation \eqref{eqn:CRTSymOps} between the transfer matrix and the translation operator. We take an alternating sum of these equalities and obtain
    \begin{equation}
      \mathcal T \left(\sum_{j=0}^{L}(-1)^j \q_j\right)+(a+b) \left(\sum_{j=0}^{L}(-1)^j \q_j\right)\mathcal T = (-1)^{L} \mathcal S^{L}\mathcal A \mathcal S^{-L} + \mathcal S^{-1} \mathcal A \mathcal S.
    \end{equation}
    The expression on the right-hand side can be simplified. Indeed, we have $\mathcal S^L=1$ on $V^L$ and $\mathcal S^{L+1}=1$ on $V^{L+1}$. Hence, we obtain $\mathcal S^{L}\mathcal A \mathcal S^{-L} = \mathcal S^{-1}\mathcal S^{L+1}\mathcal A \mathcal S^{-L}=\mathcal S^{-1}\mathcal A$. We thus find
    \begin{equation}
      \mathcal T \left(\sum_{j=0}^{L}(-1)^j \q_j\right)+(a+b) \left(\sum_{j=0}^{L}(-1)^j \q_j\right)\mathcal T = \mathcal S^{-1}\mathcal A \left((-1)^L+\mathcal S\right).
    \end{equation}
   On $W^L$ the left-hand side is, up to a factor, equal to $\mathcal T \Q+(a+b)\Q \mathcal T$. The right-hand side vanishes on $W^L$. This proves that the commutation relation \eqref{eqn:TQRelation} holds on the subspace of alternate-cyclic states.
         \end{proof}
\end{proposition}

Next, we consider operators that have a commutation relation similar to \eqref{eqn:TQRelation} with the supercharge. We show that the evaluation of their expectation values with respect to the alternate-cyclic zero-energy states can be reduced to matrix elements between the corresponding (co)homology representatives \cite{hori:03}.
\begin{proposition}
\label{prop:MatrixElementOp}
  Let $\mathcal A$ be an operator that maps $V^L$ to $V^L$ for each $L\geqslant 2$ and commutes with the supercharge according to
  \begin{equation}
  \mathcal A \Q = \lambda \Q \mathcal A ,
  \label{eqn:CommQA}
\end{equation}
where $\lambda$ is a non-zero complex number. Furthermore, let $|\Psi\rangle \in W^L$ be a zero-energy state whose decompositions \eqref{eqn:CohomRep} and \eqref{eqn:HomRep} are $|\Psi\rangle = |\Phi\rangle+\Q|\alpha\rangle$ and $|\Psi\rangle = |\Phi'\rangle+\Q^\dagger|\beta\rangle$, respectively. Then we have
\begin{equation}
  \langle \Psi|\mathcal A |\Psi\rangle =  \langle \Phi'|\mathcal A |\Phi \rangle.
  \label{eqn:MatrixElementOp}
\end{equation}
\begin{proof}
  The proof consists of a straightforward computation. First, we use \eqref{eqn:HomRep} to write
\begin{equation}
  \langle \Psi|\mathcal A |\Psi\rangle = \langle \Phi'|\mathcal A |\Psi\rangle+(\langle \beta|\Q)\mathcal A |\Psi\rangle = \langle \Phi'|\mathcal A |\Psi\rangle+\lambda^{-1}\langle \beta|\mathcal A (\Q|\Psi\rangle).
\end{equation}
Because of \eqref{eqn:QOnGS} the last term on the right-hand side vanishes. 

Second, with the help of \eqref{eqn:CohomRep} we find
\begin{equation}
  \langle \Psi|\mathcal A |\Psi\rangle=\langle \Phi'|\mathcal A |\Psi\rangle = \langle \Phi'|\mathcal A|\Phi\rangle + \langle \Phi'|\mathcal A (\Q |\alpha\rangle) = \langle \Phi'|\mathcal A|\Phi\rangle + \lambda(\langle \Phi'|\Q)\mathcal A |\alpha\rangle.
\end{equation}
Because of $\Q^\dagger|\Phi'\rangle = 0$, we conclude that the second term on the right-hand side equals zero. This leads to \eqref{eqn:MatrixElementOp}.
\end{proof}
\end{proposition}

\subsection{The eigenvalue computation}
\label{sec:EV}

We now compute the action of the transfer matrix $\mathcal T$ on the space of alternate-cyclic zero-energy states. Since $[\mathcal T, H_{\text{\tiny XYZ}}]=[\mathcal T, H]=0$, 
this space is stable under the action of $\mathcal T$. Hence, we may deduce its action from the evaluation of the matrix elements of $\mathcal T$ between the states $|\Psi_n\rangle$ and $|\bar \Psi_n\rangle$. From \cref{prop:Symmetries} and \eqref{eqn:CRTSymOps} we infer that
\begin{equation}
  \langle \Psi_n|\mathcal T|\bar \Psi_n\rangle = \langle \bar \Psi_n|\mathcal T| \Psi_n\rangle = 0.
\end{equation}
It immediately follows that both $|\Psi_n\rangle$ and $|\bar\Psi_n\rangle$ are eigenstates of the transfer matrix. To find the corresponding eigenvalues, we consider the diagonal matrix elements
\begin{equation}
  \Theta_n = \frac{\langle \bar \Psi_n|\mathcal T|\bar \Psi_n\rangle}{\langle \bar \Psi_n|\bar \Psi_n\rangle} = \frac{\langle  \Psi_n|\mathcal T| \Psi_n\rangle}{\langle \Psi_n| \Psi_n\rangle}.
  \label{eqn:EVRitz}
\end{equation}
Here, the equality of the matrix elements for $|\Psi_n\rangle$ and $|\bar \Psi_n\rangle$ is again a consequence of \cref{prop:Symmetries} and \eqref{eqn:CRTSymOps}. For $a+b \neq 0$, we apply \cref{prop:MatrixElementOp} and reduce the matrix element for $|\bar \Psi_n\rangle$ to a matrix element between representatives. Using \cref{prop:HomRep,prop:AlternativeCohomRep}, we obtain the equality
\begin{equation}
  \Theta_n =\langle \bar \Phi_n(\zeta^{-1})|\mathcal T|{\uparrow\cdots\uparrow}\rangle,
  \label{eqn:ThetaIntermediate}
\end{equation}
where $|\bar \Phi_n(\zeta)\rangle$ is the state defined in \eqref{eqn:DefPhi}.
The case where $a+b =0$ can be treated as a suitable limit of this result. 
The resulting matrix element can be explicitly computed:
\begin{proposition}
 \label{prop:Theta}
 For each $n\geqslant 1$, we have $\Theta_n = (a+b)^{2n+1}$.
 \begin{proof}
   We write $|\bar \Phi_n(\zeta^{-1})\rangle$ as a linear combination of the canonical basis states \eqref{eqn:BasisVectors}.  To this end, we introduce the notation
\begin{equation}
  ||x_1,\dots,x_{k}\rrangle = |{\uparrow\cdots\uparrow\underset{x_1}{\downarrow}\uparrow\cdots \uparrow\underset{x_2}{\downarrow}\uparrow\underset{\dots}{\cdots} \uparrow\underset{x_{k}}{\downarrow}\uparrow\cdots \uparrow}\rangle,
\end{equation}
where $k=1,\dots,L$. Using \eqref{eqn:DefPhi}, we find
\begin{equation}
  |\bar \Phi_n(\zeta^{-1})\rangle = |{\uparrow\cdots\uparrow}\rangle+\sum_{m=1}^{n}\zeta^{-m} \sum_{1\leqslant x_1<\cdots<x_{2m}\leqslant 2n+1} ||x_1,\dots,x_{2m}\rrangle.
\end{equation}
Hence, we obtain
\begin{equation}
  \Theta_n =\langle {\uparrow\cdots\uparrow}|\mathcal T|{\uparrow\cdots\uparrow}\rangle+\sum_{m=1}^n\zeta^{-m}\sum_{1\leqslant x_1<\cdots<x_{2m}\leqslant 2n+1} \llangle x_1,\dots,x_{2m}||\mathcal T|{\uparrow\cdots\uparrow}\rangle.
  \label{eqn:ThetaSum}
\end{equation}

The matrix elements on the right-hand side of this equality are readily evaluated. We have
\begin{subequations}
\begin{equation}
  \langle {\uparrow\cdots\uparrow}|\mathcal T|{\uparrow\cdots\uparrow}\rangle = a^{2n+1}+b^{2n+1}.
\end{equation}
Furthermore, for $m=1,\dots,n$ we obtain
\begin{equation}
  \llangle x_1,\dots,x_{2m}||\mathcal T|{\uparrow \cdots \uparrow}\rangle = \zeta^m \left(\alpha(x_1,\dots,x_{2m})+\delta(x_1,\dots,x_{2m})\right),
\end{equation}
with
\begin{align}
  \alpha(x_1,\dots,x_{2m}) &= a^{x_2-x_1}b^{x_3-x_2}\cdots a^{x_{2m}-x_{2m-1}}b^{2n+1-(x_{2m}-x_1)},\\
  \delta(x_1,\dots,x_{2m}) &= b^{x_2-x_1}a^{x_3-x_2}\cdots b^{x_{2m}-x_{2m-1}}a^{2n+1-(x_{2m}-x_1)}.
\end{align}%
\label{eqn:MatrixElements8VTM}%
\end{subequations}%
We substitute these expressions into \eqref{eqn:ThetaSum} and find
\begin{equation}
  \Theta_n = a^{2n+1}+b^{2n+1} +\sum_{m=1}^n\sum_{1\leqslant x_1<\cdots<x_{2m}\leqslant 2n+1} (\alpha(x_1,\dots,x_{2m})+\delta(x_1,\dots,x_{2m})).
  \label{eqn:ThetaSum2}
\end{equation}

The evaluation of this sum reduces to a combinatorial problem. To see this, we consider the set of all words $\gamma=(\gamma_1,\dots,\gamma_{2n+1})$ of length $2n+1$ with letters $\gamma_j\in \{a,b\}$. We assign a weight $\omega(\gamma) = \gamma_1\gamma_2\cdots \gamma_{2n+1}$ to each word $\gamma$. Two simple examples are $\gamma=(a,a,\dots,a)$ and $\gamma=(b,b,\dots,b)$ whose weights are $\omega(\gamma)=a^{2n+1}$ and $\omega(\gamma)=b^{2n+1}$, respectively. Every other word contains both letters $a$ and $b$. For each such word there is an integer $m=1,\dots, n$ and a sequence of integers $1\leqslant x_1 < x_2 < \cdots < x_{2m}\leqslant 2n+1$ such that either
 \begin{align*}
   \gamma = (b,\dots,b,\underset{x_1}{a},\dots,a,\underset{x_2}{b},\dots,b,\dots,\underset{x_{2m-1}}{a},\dots,a,\underset{x_{2m}}{b},\dots,b), \quad \omega(\gamma) = \alpha(x_1,\dots,x_{2m}),
\end{align*}
or
\begin{align*}
   \gamma = (a,\dots,a,\underset{x_1}{b},\dots,b,\underset{x_2}{a},\dots,a,\dots,\underset{x_{2m-1}}{b},\dots,b,\underset{x_{2m}}{a},\dots,a), \quad \omega(\gamma) = \delta(x_1,\dots,x_{2m}).
 \end{align*}
 We conclude that the sum \eqref{eqn:ThetaSum2} can be written as a sum over all words $\gamma$. The terms to sum up are their corresponding weights $\omega(\gamma)$. Hence we find
 \begin{equation}
  \Theta_n = \sum_{\gamma_1 = a,b}\cdots \sum_{\gamma_{2n+1}=a,b} \gamma_1\cdots \gamma_{2n+1} = (a+b)^{2n+1}.
\end{equation}
This concludes the proof.
 \end{proof}
\end{proposition}

\begin{proof}[Proof of \cref{thm:MainResult}]
According to \cref{prop:Restriction} every eigenstate of the transfer matrix of the supersymmetric eight-vertex model with the eigenvalue $\Theta_n = (a+b)^{2n+1}$ is an alternate-cyclic zero-energy state of the Hamiltonian $H$. The space of these states is spanned by $|\Psi_n\rangle$ and $|\bar \Psi_n\rangle$. According to \cref{prop:Theta}, we have
\begin{equation}
  \mathcal T|\Psi_n\rangle = (a+b)^{2n+1}|\Psi_n\rangle, \quad \mathcal T|\bar \Psi_n\rangle = (a+b)^{2n+1}|\bar \Psi_n\rangle.
\end{equation}
Hence, the eigenspace of $\Theta_n$ is two dimensional and spanned by $|\Psi_n\rangle$ and $|\bar\Psi_n\rangle$. Furthermore, according to \cref{thm:GS} it is the space of ground states of the XYZ Hamiltonian with the anisotropy parameters \eqref{eqn:CL}. The ground-state eigenvalue is $E_0=-(2n+1)(3+\zeta^2)/4$.
\end{proof}
Throughout this article, we have considered non-zero vertex weights. However, for $L=2n+1$ the transfer matrix still possesses the eigenvalue $\Theta_n$ if some of the vertex weights are zero. The reason is that the eigenvalues are continuous functions of the entries of $\mathcal T$. Hence, they are continuous with respect to $a,b,c,d$. In particular, $\Theta_n$ is still a transfer-matrix eigenvalue as the vertex weight $d$ tends to zero. In this limit, the supersymmetric eight-vertex model reduces to the six-vertex model with the anisotropy parameter $\Delta=-1/2$. In this special case, the existence of the eigenvalue can be proven by other techniques \cite{razumov:07}.

We conclude this section by proving that if the vertex weights are positive then $\Theta_n$ is the largest eigenvalue of the transfer matrix:
\begin{theorem}
 If the vertex weights are positive, $a,b,c,d>0$, then for each $L=2n+1,$ $n\geqslant 1$, $\Theta_n=(a+b)^{2n+1}$ is the largest eigenvalue of the transfer matrix of the supersymmetric eight-vertex model.
\begin{proof}
We denote by $\mathcal T_{\pm}$ the restriction of the transfer matrix to the eigenspace of the spin-parity operator with eigenvalue $\pm 1$. The matrix elements of $\mathcal T_\pm$ with respect to the canonical basis of this eigenspace can be explicitly computed \cite{baxterbook}. If $a,b,c,d>0$ then these matrix elements are  positive and thus $\mathcal T_\pm$ is a positive matrix.   It follows from the Perron-Frobenius theorem \cite{meyer:00} that $\mathcal T_\pm$ has a largest positive eigenvalue $\Theta_\pm$ which is non-degenerate. There is a unique vector $|\Phi_\pm\rangle$ with positive components and norm $1$ such that
  \begin{equation}
    \mathcal T_\pm|\Phi_\pm\rangle = \Theta_\pm |\Phi_\pm\rangle .
  \end{equation}
Furthermore, except for positive multiples of $|\Phi_\pm\rangle$, the matrix $\mathcal T_\pm$ has no other eigenstate with positive components.
  
Let us consider the state $\ket{\Psi_\pm}$ defined in the proof of \cref{thm:GS}. It has positive components and norm one. Furthermore, according to \cref{thm:MainResult} it is an eigenstate of $\mathcal T_\pm$:
\begin{equation}
\mathcal T_\pm \ket{\Psi_\pm} = \Theta_n \ket{\Psi_\pm}.
\end{equation}
By the uniqueness of $\ket{\Phi_\pm}$, we conclude that $\ket{\Phi_\pm} = |\Psi_\pm\rangle$ and consequently, $\Theta_\pm = \Theta_n$. 
\end{proof}
\end{theorem}

We notice that this result immediately implies that for positive vertex weights the free energy
 per site of the supersymmetric eight-vertex model is given by
\begin{equation}
  f = -\lim_{n\to \infty} \frac{1}{2n+1}\ln \Theta_n = -\ln(a+b).
\end{equation}
As expected, this agrees with Baxter's result for the free energy per site of the eight-vertex model \cite{baxterbook}, if specialised to the supersymmetric case.

\section{Conclusion}
\label{sec:Conclusion}
In this article, we have proven Stroganov's conjecture \cite{stroganov:01,stroganov:01_2} about the transfer matrix of the supersymmetric eight-vertex model on the square lattice. It states that this transfer matrix with $L=2n+1$ vertical lines and periodic boundary conditions possesses the simple eigenvalue $\Theta_n = (a+b)^{2n+1}$. Furthermore, we have shown that the corresponding eigenspace is spanned by the ground states of a related XYZ Hamiltonian. The proof relies on the lattice supersymmetry of this Hamiltonian. An important ingredient of the proof is the commutation relation between the supercharge and the transfer matrix. This result is interesting in itself as it illustrates a relation between lattice supersymmetry and quantum integrability. We have used it to reduce the computation of expectation values of the transfer matrix with respect to the spin-chain ground states to the evaluation of simple matrix elements.

Let us briefly discuss two generalisations of the present work. First, an obvious generalisation is the extension of the present results to other boundary conditions. In a companion paper \cite{hagendorf:17_3}, we consider the supersymmetric eight-vertex model on a strip and the related XYZ chain with open boundary conditions.
We classify all boundary terms that are compatible with a lattice supersymmetry. Furthermore, we determine the subset of boundary terms for which supersymmetry singlets exist. In the case of open spin chains these are the Hamiltonian's ground states \cite{hagendorf:17} (without any restriction to a particular subsector of the spin-chain Hilbert space). Adapting the strategy of the present article, we compute the action of the transfer matrix on the supersymmetry singlets and determine the corresponding transfer-matrix eigenvalue.
Second, one may consider the transfer matrix of the inhomogeneous eight-vertex model with vertex weights that locally fulfil \eqref{eqn:CL8V}. In the case of periodic boundary conditions with $L=2n+1$, a simple and explicit expression of a doubly degenerate transfer-matrix eigenvalue that reduces to $(a+b)^{2n+1}$ in the homogeneous limit has been conjectured \cite{razumov:10}. This conjecture allows one to determine interesting properties of the states $|\Psi_n\rangle$ and $|\bar \Psi_n\rangle$ such as sum rules \cite{zinnjustin:13} and special components \cite{hagendorf:17_4}. A proof of this conjecture would therefore be interesting.
   
\subsubsection*{Acknowledgements}
   This work is supported by the Belgian Interuniversity Attraction Poles Program P7/18 through the network DYGEST (Dynamics, Geometry and Statistical Physics). C.H. acknowledges hospitality and support from the Mathematical Research Institute \textsc{Matrix} and from the program ``Combinatorics, statistical mechanics and conformal field theory'', where this work was completed. We are very grateful to Alexi Morin-Duchesne for his interesting remarks and stimulating suggestions that helped us to improve our manuscript. Furthermore, we would like to thank Robert Weston for an interesting discussion.

\end{document}